\documentclass[12pt,journal,draftclsnofoot,letterpaper,onecolumn]{IEEEtran} 
\usepackage{url}          
\usepackage[cmex10]{amsmath}
\usepackage{amsfonts}
\usepackage{amssymb}
\usepackage{cite}
\usepackage[caption=false,font=footnotesize]{subfig}

\ifCLASSINFOpdf
  \usepackage[pdftex]{graphicx}
  \DeclareGraphicsExtensions{.pdf,.jpeg,.jpg,.png}
\else
  \usepackage[dvips]{graphicx}
  \DeclareGraphicsExtensions{.eps}
\fi

\newtheorem{theorem}{Theorem}
\newtheorem{lemma}[theorem]{Lemma}

\begin{document}

\title{ A Distributed Differential Space-Time Coding Scheme With Analog Network Coding in Two-Way Relay Networks}

\author{Qiang~Huo, Lingyang~Song, Yonghui~Li, and~Bingli~Jiao%
\thanks{Manuscript received  September 11, 2011; revised January 18, 2012 and May 7, 2012; accepted May 16, 2012.
Date of publication June 05, 2012; date of current version August 07, 2012.
The associate editor coordinating the review of this manuscript and approving it for publication was Prof. Yao-Win (Peter) Hong.
The work of L.~Song and B.~Jiao is   partially supported by the National Natural Science
Foundation of China 60972009 and 61061130561, National Science and
Technology Major Projects of China 2011ZX03005-003-02, and by Specialized Research Fund for the Doctoral
Program of Higher Education of China 20110001110102.
The work of Y.~Li is supported by the Australian Research Council Discovery Projects DP0985140, DP0877090, and  DP120100190, and Linkage project LP0991663.}%
\thanks{Q.~Huo, L.~Song, and B.~Jiao are with the State Key
Laboratory of Advanced Optical Communication Systems and Networks, School of Electronics Engineering and Computer Science, Peking University, Beijing, China, 100871 (e-mail: qiang.huo@pku.edu.cn; lingyang.song@pku.edu.cn; jiaobl@pku.edu.cn).}%
\thanks{Y.~Li is with the University of Sydney, Sydney, NSW 2006, Australia (e-mail: yonghui.li@sydney.edu.au).}%
\thanks{Color versions of one or more of the figures in this correspondence are available
online at http://ieeexplore.ieee.org.}%
\thanks{Digital Object Identifier 10.1109/TSP.2012.2202654}%
}

\markboth{IEEE TRANSACTIONS ON SIGNAL PROCESSING,~Vol.~60, No.~9, SEPTEMBER~2012}%
{Huo \MakeLowercase{\textit{et al.}}: A Distributed Differential Space-Time Coding Scheme With Analog Network Coding in Two-Way Relay Networks}



\maketitle



\begin{abstract}
In this paper, we consider general two-way relay networks~(TWRNs) with two source and $N$ relay  nodes.
A distributed differential space time coding with analog network coding~(DDSTC-ANC)  scheme  is proposed.
A simple blind estimation and a differential signal detector are developed to recover the desired signal at each source.
The pairwise error probability~(PEP) and  block error rate~(BLER)  of the  DDSTC-ANC scheme are analyzed. Exact  and simplified PEP  expressions are derived.
To improve the system performance, the optimum power allocation~(OPA) between the source and relay nodes is determined based on the simplified PEP expression.
The analytical results are verified through simulations.
\end{abstract}

\newpage


\IEEEpeerreviewmaketitle




\section{Introduction}\label{sec:introduction}
\IEEEPARstart{I}{t} is well known that cooperative communication improves system robustness and capacity by allowing nodes to cooperate in their transmission to form a virtual antenna array~\cite{Laneman2003P2415}.
Compared to one-way relay networks~(OWRN), two-way communication is an effective scheme to improve the spectral efficiency by allowing the simultaneous exchange of two-way information flows.
In \cite{Rankov2006P1668}, the authors first studied the two-way relay networks~(TWRN) and derived its achievable bidirectional  rate.
The TWRNs have attracted  increased interest  due to its high spectral efficiency. Various protocols for the TWRNs have been proposed recently \cite{Popovski2007P16,Yuen2008P1385}.

In \cite{Yuen2008P1385}, the conventional network coding scheme was applied to the TWRNs. In this scheme,    two source nodes transmit signals to the relay, separately. The
relay decodes the received signals, performs binary network coding,
and   broadcasts network coded symbols back to both source
nodes.
However, this scheme may cause irreducible error floor due to
the detection errors which occur at the relay node.
In \cite{Popovski2007P16}, an amplify and forward based network coding scheme  
was proposed.
In this scheme, both source nodes transmit at the same time so that the relay receives a superimposed signal.
The relay   amplifies the received signal, and broadcasts it to both source nodes.
Each source node  subtracts its own contribution and estimates the signal transmitted from the other source node.
Analog network coding is particularly
useful in wireless networks as the wireless channel acts as a natural
implementation of network coding by summing the wireless signals
over the air.

Recently, distributed space-time coding  for  OWRNs was proposed in \cite{Jing2006P3524}  to achieve spatial diversity.
Since OWRNs take place only in a single-direction, to further improve the spectral efficiency of the   relay networks, the distributed space-time coding was proposed for TWRNs in \cite{Cui2009P658} and \cite{Wang2010P5331}.
However, most of the existing works on distributed space-time coding in TWRNs consider coherent detection at each receiver with the assumption of available  channel-state information~(CSI).
In some situations, e.g.,
the fast-fading environment, the acquisition of accurate CSI
presents great challenge, and   training becomes expensive and inefficient while there are  a large number of relays in the wireless networks\cite{Song2010P3933}.
In this case, differential modulation would be a practical solution because it requires no knowledge of the CSI.

The distributed differential space-time coding was first proposed for OWRNs in \cite{Jing2008P1092}.
In TWRNs, the signal received at the relay node is a superposition of two symbols sent from two source nodes.
Thus, if there is no CSI available at source and relay nodes, it will be very difficult to design distributed differential modulation schemes in TWRNs.
The challenge is due to the blind channel estimation from the superimposed signals at the relay and unknown self-interference at each destination.
In \cite{Utkovski2009P779}, the authors first extended  the distributed differential space-time coding to TWRNs. In order to enable  differential encoding
and decoding, this scheme starts with a four-stage initialization phase,  which is similar to traditional one-way relaying,  to transmit the bi-directional reference signals respectively. After initialization, each user then proceeds to the data transmission. Information exchange between two users is done in two time slots. However, the decoding algorithm in \cite{Utkovski2009P779} is a noncoherent detection scheme where the decoding of current symbol is based on the estimation of the previous symbol. Consequently, when one symbol was decoded incorrectly, it will affect the decoding of   consecutive symbols thus leading to serious error propagation. To solve this problem, periodical initialization of the protocol has to be performed to transmit new reference signals for decoding, making the proposed scheme inefficient. Furthermore, no pairwise error probability~(PEP) analysis was performed in \cite{Utkovski2009P779}  due to the complexity of the protocol.
 Song \textit{et al.} \cite{Song2010P3933}  presented an analog network coding scheme with differential modulation using the amplify-and-forward protocol for bidirectional relay networks.
However, this scheme is limited to single relay node, thus cannot be extended to the distributed space-time codes.

Unlike \cite{Jing2008P1092,Utkovski2009P779,Song2010P3933}, in this paper, we propose a distributed differential space time coding with analog network coding~(DDSTC-ANC)  scheme for the TWRNs with multiple relays.  
In this scheme, two source nodes perform differential modulation,
and transmit the differential modulated symbols to all the relay nodes in the first time slot.
The signal received at the relay node is a superposition of two transmitted symbols.
In the second time slot, the $N$ relay nodes broadcast the processed signals  to both source nodes simultaneously.
We propose a blind estimation technique that can be used to subtract the self-interference without knowledge of CSI at both relay nodes and two source nodes. A  simple differential signal detector is then developed to recover the desired signal at each source.
The performance of the proposed differential DDSTC-ANC scheme is analyzed and the PEP  and block error rate~(BLER) expressions are derived.
They show  that the proposed differential scheme can achieve the same diversity order as the coherent detection scheme  but is about $3$dB away compared to the coherent detection scheme due to the differential transmission.
To further improve the system performance, the optimum power allocation~(OPA) between the source nodes and the relay nodes is determined based on the provided simplified PEP expression.
The analytical results are verified through simulations.
Simulation results also show that the proposed differential scheme with OPA yields 
superior performance improvement over an equal power allocation~(EPA) scheme.

The rest of this paper is organized as follows:
In Section \ref{sec:system_model}, the system model is introduced. Section \ref{sec:DD-STC-ANC} presents the proposed DDSTC-ANC scheme.
The performance and diversity order of  DDSTC-ANC  are analyzed in Section \ref{sec:PEP_FER}.
In Section \ref{sec:OPA}, the OPA  for the DDSTC-ANC is presented.
Simulation results are provided in Section \ref{sec:simulations}.
In Section \ref{sec:conclusion}, we draw the main conclusions.

\emph{\textbf{Notation}}:
Matrices and vectors are denoted using capital letters and boldface lowercase letters, respectively.
$(\cdot)^*$, $(\cdot)^T$ and $(\cdot)^H$  represent conjugate, transpose  and conjugate transpose, respectively, for both matrix and vector.
For a complex matrix $A$, $\det A$  denotes the determinant
  of A. $\mathbf{I}_m$ is the $m \times m$ identity
matrix.
$\text{diag}\{a_1, \cdots, a_n\}$ stands for an $n \times n$ diagonal
matrix whose $i$th diagonal entry is $a_i$.
$\ln$ represents the natural logarithm, and $||\cdot||$  is the Frobenius norm.
$\mathbb{E}$ and $ P(\cdot)$ denote the expectation
and probability, respectively.



%
%
%



\section{System Model}\label{sec:system_model}

\begin{figure}[e] 
\centering
\graphicspath{{fig/}}
\includegraphics[width=0.75 \textwidth]{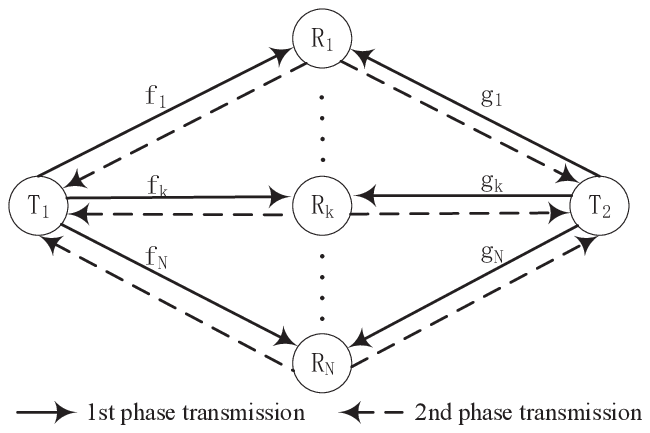}
\caption{Block diagram of the two-hop TWRN.}  \label{fig:SystemDiagram}
\end{figure}

%

In this paper, we consider a  general TWRN  with $N+2$ nodes,  as shown in Fig. \ref{fig:SystemDiagram},
where two source nodes, $T_1$
and $T_2$, want to exchange information with each other
through $N$
relay nodes.
It is assumed that
each node in the network is equipped
with one single antenna working in the half-duplex mode. 
We consider a quasi-static fading channel, where the channel remains constant for the duration of a
frame
and varies independently from one
frame to another.
Let $f_i$ and $g_i$ denote the complex fading channel coefficients of $T_1 - R_i$ and $T_2 - R_i$, respectively.
Furthermore, we assume Rayleigh flat fading channels, i.e., $f_i\sim\mathcal{CN}(0, \sigma_{f_i}^{2})$ and  $g_i\sim\mathcal{CN}(0, \sigma_{g_i}^{2})$, respectively.
For analysis tractability, symmetry of
the relay nodes is assumed in this paper, i.e.,
$\sigma_{f_i}=\sigma_{f},\forall i$ and $\sigma_{g_i}=\sigma_{g},\forall i$.

A general two-time slot TWRN protocol is used, as shown in Fig.~\ref{fig:SystemDiagram}.
In the first time slot, both $T_1$
and $T_2$ transmit their messages
and the relays $\{R_1, \cdots, R_N\}$ receive a superposition of the
signals transmitted from $T_1$
and $T_2$.
Let $\mathbf{s}(t)=[s_1(t),\cdots,s_T(t)]^T$ and $\mathbf{d}(t)=[d_1(t),\cdots,d_T(t)]^T$ denote the transmitted symbol vectors of  $T_1$  and  $T_2$ at time $t$, respectively.
They are normalized as $\mathbb{E}\{ \mathbf{s}(t) \mathbf{s}(t)^H\}=
    \mathbb{E}\{ \mathbf{d}(t) \mathbf{d}(t)^H\}=\mathbf{I}_T.$
The received signal vector at $R_i$  can be written as
\begin{equation}\label{ }
    \mathbf{r}_i(t)=\sqrt{P_1 }f_i(t)\, \mathbf{s}(t) + \sqrt{P_2 }g_i(t)\, \mathbf{d}(t) + \mathbf{v}_i(t),
\end{equation}
where $P_1$ and $P_2$ denote the transmit power of $T_1$ and $T_2$, respectively, and $\mathbf{v}_i(t)$ represents the  noise vector at $R_i$
and each noise term follows a zero-mean  complex additive white Gaussian distribution, i.e., $\mathbf{v}_i(t)\sim\mathcal{CN}(0,  N_0 \mathbf{I}_T)$.

During the second time slot, $R_i$ processes $\mathbf{r}_i(t)$ to generate a space time coded symbol vector $\mathbf{x}_i(t)$.
In this paper, we consider  the amplify-and-forward protocol in the relay nodes. The transmit signal at the $i$th relay is designed to be a linear function of its received signal and its conjugate\cite{Hassibi2002P1804}:
\begin{equation}\label{eq:x_Ri}
\begin{split}
  &\mathbf{x}_i(t)= \beta_i(t)\left(A_i \mathbf{r}_i(t) +B_i \mathbf{r}_i(t)^{*}\right),
\end{split}
\end{equation}
where $A_i$ and  $B_i$ are two
$T\times T$
complex matrices specifically
designed for the construction of distributed space-time codings, and $\beta_i(t)$ is
the scaling factor at $R_i$.

In this work, the scaling factor $\beta_i(t)$ in Eq. (\ref{eq:x_Ri}) can be obtained based on
the available statistical CSI, which is specifically given by\cite{Maham2009P2036}
 \begin{equation}\label{ }
\begin{split}
 \beta_i (t)=\sqrt{\frac{P_{R_i}}{\sigma_{f_i}^{2}P_1 +\sigma_{g_i}^{2}P_2 +N_0}},
\end{split}
\end{equation}
 where $P_{R_i}$ is the transmitted power of $R_i$. Since we assume $P_{R_i}=P_R$, we have   $\beta_i (t)=\beta, \forall i$.
For simplicity, in this paper, we only design the system  that either
$A_i$  is unitary, $B_i=\mathbf{0}_T$  (case I) or $B_i$  is unitary, $A_i=\mathbf{0}_T$ (case II). Thus, case I means that the $i$th column of the code matrix ($S(t)$ and $D(t)$ in Eq. (\ref{eq:y_2_special})) contains only the
transmitted symbols, and  case II  means that the $i$th column of the code matrix contains the linear combinations of the conjugate of the transmitted symbols only.
Further more, we assume that $T=N$, i.e., the number of symbols in a space-time block code is equal to the number of  relay nodes.
We further define
\begin{equation}\label{eq:notation_define}
\begin{split}
\begin{cases}
    & O_i \triangleq A_i,\; \hat{f}_i \triangleq {f}_i,\; \hat{g}_i \triangleq {g}_i,\;
      \hat{\mathbf{v}}_i(t)\triangleq \mathbf{v}_i(t), \\
     &\; \; \;  \hat{\mathbf{s}}_i(t) \triangleq \mathbf{s}(t), \; \hat{\mathbf{d}}_i(t) \triangleq \mathbf{d}(t), \;\;\;\;\;\;\;\;\;\;\;\;\;\; \;\;\;\; \; \; \;\; \text{if } B_i=\mathbf{0}_T,   \\   
    & O_i \triangleq B_i,\; \hat{f}_i \triangleq {f}_i^*,\; \hat{g}_i \triangleq {g}_i^*,\;
     \hat{\mathbf{v}}_i(t)\triangleq \mathbf{v}_i(t)^*, \\
     &\; \; \;   \hat{\mathbf{s}}_i(t) \triangleq \mathbf{s}(t)^*, \; \hat{\mathbf{d}}_i(t) \triangleq \mathbf{d}(t)^*, \;\;\;\;\;\;\;\;\;\;\;\;\;\; \;\;\;\; \; \text{if } A_i=\mathbf{0}_T. \\    
\end{cases}
\end{split}
\end{equation}

Then the relay node $R_i$ broadcasts the coded symbol
vector $\mathbf{x}_i(t)$ back to both source nodes.
Since $T_1$ and $T_2$ are mathematically symmetrical,
for simplicity, in the following, we only discuss the decoding
and the analysis for the signals received by $T_2$\cite{Song2010P2932}.
The received signal vectors at  $T_2$  is given by 
\begin{equation}\label{eq:y_2}
    \mathbf{y}_2(t)=\sum_{i=1}^{N} g_i(t)\mathbf{x}_i(t)+ \mathbf{w}_2(t),
\end{equation}
 where $\mathbf{w}_2(t)$
denotes the independent and identically  distributed~(i.i.d)  additive white Gaussian noise~(AWGN) vectors  at $T_2$, and we have  $ \mathbf{w}_2(t) \sim\mathcal{CN}(0, N_0 \mathbf{I}_{T}) $.

The received signal at $T_2$ can then be rewritten as:
\begin{equation}\label{eq:y_2_special}
\begin{split}
    \mathbf{y}_2(t)&= \sqrt{P_1 } \,S(t) \mathbf{h}_{12}(t)+\sqrt{P_2 } \,D(t) \mathbf{h}_{22}(t) + \mathbf{n}_2(t),
\end{split}
\end{equation}
where
\begin{equation}\label{}
\begin{split}
\begin{cases}
  &S(t)=\left[ O_1\hat{\mathbf{s}}_1 (t), \cdots,  O_N\hat{\mathbf{s}}_N(t)\right],\\
  &D(t)=[ O_1\hat{\mathbf{d}}_1 (t), \cdots,  O_N\hat{\mathbf{d}}_N (t)],\\
  &\mathbf{h}_{12}(t) = [\beta_1(t)\hat{f}_1(t)g_1(t),$ $ \cdots, \beta_N(t)\hat{f}_N(t)g_N(t)]^T,\\
  &\mathbf{h}_{22}(t) = [\beta_1(t)\hat{g}_1(t)g_1(t), \cdots, \beta_N(t)\hat{g}_N(t)g_N(t)]^T,\\
  &\mathbf{n}_2(t)= \sum_{i=1}^{N}  \beta_i(t)$ $g_i(t) O_i \hat{\mathbf{v}}_i(t)   + \mathbf{w}_2(t).
  \end{cases}
\end{split}
\end{equation}
It is easy to prove that
$    \mathbb{E}\{ \mathbf{n}_2(t) \mathbf{n}_2(t)^H \}
      =\sigma_{\mathbf{n}_2}^2(t)\mathbf{I}_N
$
, and $\sigma_{\mathbf{n}_2}^2(t)= ( \sum_{i=1}^{N} |\beta_i(t)|^2 $ $|g_i(t)|^2  +1 )  N_0$.

\section{Distributed Differential Space-Time Coding For TWRNs}\label{sec:DD-STC-ANC}
In this section, we propose a  distributed  differential scheme. 
First, we blindly estimate channel $\mathbf{h}_{22}(t)$ defined in
Eq. (\ref{eq:y_2_special}), which can be used to subtract the self-interference. Then,  a simple differential signal detector is developed to recover the desired signal at source $T_2$.

In the proposed DDSTC-ANC, $T_1$ encodes a message  at time $t$ into an $N \times N$ unitary matrix $U(t)$, which is then differentially encoded as $\mathbf{s}(t) = U(t)\cdot \mathbf{s}(t-1),$
where $\mathbf{s}(t-1)$ is the signal transmitted by  $T_1$  at time $t-1$. Similarly, $T_2$ differentially encodes a message  at time $t$  into an $N \times N$ unitary matrix $V(t)$, which is then differentially encoded as $\mathbf{d}(t) = V(t)\cdot \mathbf{d}(t-1).$

For the first block, we can transmit a known vector to both source nodes  that satisfies
$
 \mathbb{E}\{  \mathbf{s}(t)^H \mathbf{s}(t) \}=
    \mathbb{E}\{ \mathbf{d}(t)^H \mathbf{d}(t) \}= N,
$
for example, $[1\; 1\;  \cdots \; 1]^T$ or $ [\sqrt{N} \; 0 \; \cdots\;  0]^T$. Similar to the differential space-time coding for multiple-antenna systems, having $U(t)$ and $V(t)$ unitary preserves the transmit power.

For simplicity, we define $\hat{U}_i(t) \triangleq {U}(t)$ if $B_i=\mathbf{0}_T$,  and $\hat{U}_i(t) \triangleq {U}(t)^{*}$ if $A_i=\mathbf{0}_T$. In the distributed  differential scheme, the codes $U(t)$ and $V(t)$ should commute with the relay matrices\cite{Jing2008P1092},  i.e., \footnote{More properties about the differential space-time
coding can be found in\cite{Hughes2000P2567,Hochwald2000P2041,Tarokh2000P1169,Jafarkhani2005P}.}
\begin{equation*}\label{eq:U_t}
\begin{split}
  O_i \hat{U}_i(t)= U(t) O_i,
\end{split}
\end{equation*}
or equivalently,
\begin{equation}\label{eq:U_t_2}
\begin{split}
\begin{cases}
  &A_i  U(t)= U(t) A_i, \;\; \text{if } B_i=\mathbf{0}_T,  \\   
  &B_i  U^*(t)= U(t)  B_i, \;\; \text{if } A_i=\mathbf{0}_T.  \\   
  \end{cases}
\end{split}
\end{equation}
Hence,
$S(t)$ can be rewritten as
\begin{equation}\label{eq:s_t_dif}
\begin{split}
   S(t) &=\left[ O_1 \hat{U}_1(t) \hat{\mathbf{s}}_1 (t-1), \cdots,  O_N \hat{U}_N(t) \hat{\mathbf{s}}_N (t-1)\right]\\
&=U(t)\cdot \left( O_1  \hat{\mathbf{s}}_1 (t-1), \cdots,  O_N \hat{\mathbf{s}}_N (t-1)\right)\\
 &=U(t)\cdot S(t-1).
\end{split}
\end{equation}
Similarly, we have $ D(t) = V(t)\cdot D(t-1).$

The distributed differential space-time codes~(STC)  for TWRNs should be designed to satisfy  Eq. (\ref{eq:U_t_2}).
The design and choice of appropriate codes
is beyond the scope of this work,
here, we only briefly introduce
some existing  STCs that can be used in TWRNs.
For the TWRNs with two relays, we can use Alamouti code
\cite{Alamouti1998P1451}, which has full diversity and linear decoding complexity.
Square real orthogonal codes~(SORCs), which  have full diversity and linear decoding complexity, were proposed in \cite{Jing2008P1092} for two, four and eight  antennas systems.

\begin{theorem}\label{theorem:h_22}
If the relay matrices have the property:
 $\text{tr}\{ O_j O_i^H\}=N$ for $i=j$,
 $\text{tr}\{ O_j O_i^H\}=0$ for $i\neq j$, we have
\begin{equation}\label{ }
\begin{split}
    \mathbb{E}\{ D(t)^H \mathbf{y}_2(t)\} &= \sqrt{P_2 }  N\, \mathbf{h}_{22}(t),
\end{split}
\end{equation}
and $\mathbf{h}_{22}(t)$ can be approximated  as
\begin{equation}\label{eq:h_22_aprox}
\begin{split}
\mathbf{h}_{22}(t)  \approx
\frac{1}{N L}\frac{1}{\sqrt{P_2 }} \sum_{l=1}^{L}  D(t-l)^H \mathbf{y}_2(t-l),
\end{split}
\end{equation}\\
where $L$ denotes the number of STC symbols in a frame.
\end{theorem}
\begin{proof}[Proof]
It can be proved by direct matrix multiplication and expectation. Due to the limited space, we omit the details.
\end{proof}

We note that since receiver $T_2$ knows the symbols $\mathbf{d}(t)$ sent by itself, using the blindly estimated channel $\mathbf{h}_{22}(t)$, we can subtract the self-interference at $T_2$  without using pilot symbols at the beginning.
 %
Although we can  blindly estimate channel $\mathbf{h}_{22}(t)$, $T_2$ does not have any CSI of $\mathbf{h}_{12}(t)$. Then based on the above theorem, a simple differential signal detector is developed to recover the desired signal $\mathbf{s}(t)$ at source $T_2$.
In the later performance analysis section, we assume that $\mathbf{h}_{22}(t)$ is perfectly cancelled.
Most of papers on distributed STCs for TWRNs also assume perfect self-interference cancellation, such as \cite{Wang2010P5331} and \cite{Song2011P1954} for coherent systems and \cite{Song2010P3933} and \cite{Song2010P2932} for differential systems.
However, in practice, the estimation error will introduce some performance degradation which depends on estimation accuracy of $\mathbf{h}_{12}(t)$.
The estimated $\mathbf{h}_{22}(t)$ is used in simulations in this paper. In the simulation section, we have simulated the proposed scheme using the estimated $\mathbf{h}_{22}(t)$  and the results show that  the performance loss due to the $\mathbf{h}_{22}(t)$ estimation error is negligible.

By using Eq. (\ref{eq:s_t_dif}) and Eq. (\ref{eq:h_22_aprox}) and the assumption of $\mathbf{h}_{12}(t)=\mathbf{h}_{12}(t-1)$, we have
\begin{equation}\label{eq:y_2_dif}
\begin{split}
    \tilde{\mathbf{y}}_2(t)
&=\mathbf{y}_2(t)- \sqrt{P_2 } \,D(t)\mathbf{h}_{22} (t) \\
&=\sqrt{P_1 } \,S(t) \mathbf{h}_{12}(t) + \mathbf{n}_2(t) \\
&=   U(t)\tilde{\mathbf{y}}_2(t-1)  +\tilde{\mathbf{n}}_2(t),
\end{split}
\end{equation}
where $ \tilde{\mathbf{n}}_2(t)= \mathbf{n}_2(t)-U(t)\mathbf{n}_2(t-1)$.
Note that $\mathbb{E}\{ U(t) U(t)^H \} = \mathbf{I}_N$, and $\mathbf{n}_2(t)$ and $\mathbf{n}_2(t-1)$  are independent complex Gaussian random vectors
with zero mean and covariance $\sigma_{\mathbf{n }_2}^2(t)$.
We have
$
\mathbb{E}\{ \tilde{\mathbf{n}}_2(t) \tilde{\mathbf{n}}_2(t)^H \}
= \sigma_{\tilde{\mathbf{n}}_2}^2(t) \mathbf{I}_T
$,
where
$\sigma_{\tilde{\mathbf{n}}_2}^2(t) =2( \sum_{i=1}^{N} |\beta_i(t)|^2 |g_i(t)|^2 +1 )  N_0
$.
Thus,  $\tilde{\mathbf{n}}_2(t)$
is a Gaussian random vector with zero mean and covariance $\sigma_{\tilde{\mathbf{n}}_2}^2(t)$.

Hence, the least square~(LS) decoder can be performed to recover the transmitted signal
\begin{equation}\label{ }
\begin{split}
  \arg\min_{U_k(t)}\|\tilde{\mathbf{y}}_2(t)-U_k(t)\tilde{\mathbf{y}}_2(t-1)\|.
\end{split}
\end{equation}



\section{Pairwise Error Probability and Block Error Rate Analysis}\label{sec:PEP_FER}
In this section, we derive the PEP and the BLER of the proposed DDSTC-ANC scheme. Asymptotic diversity order is also  analyzed in this section.
\subsection{Pairwise Error Probability}\label{subsec:PEP}
For simplicity, we define
$
 U_{\Delta,kj}(t) = U_k(t)- U_j(t)
$
and
$
S_{\Delta,kj}(t) = S_k(t)- S_j(t)
$.
The PEP  of mistaking the $k$th STC block by the $j$th STC block  can be evaluated by averaging the conditional PEP over the channel statistics, i.e., $f_i,g_i$, and we have\footnote{The superscript ``d" denotes differential scheme and ``c" represents coherent scheme.}\cite{Tse2005P}
\begin{equation}\label{eq:Q_function_PEP}
\begin{split}
P_{kj}^{d}(\gamma)
&=\mathbb{E}_{f_i,g_i}\left[Q\left(\sqrt{\frac{\|U_{\Delta,kj}(t) \tilde{\mathbf{y}}_2(t-1)\|^2}{2\sigma_{\tilde{\mathbf{n}}_2}^2(t) }}\right)\right],
\end{split}
\end{equation}
where $\gamma=\frac{P}{N_0}$ is signal-to-noise ratio~(SNR), $P$ is the total power in the TWRN
and $Q(x)=\frac{1}{\sqrt{2\pi}}\int_x^{\infty}\exp(-\frac{t^2}{2})\text{d}t$ is the Gaussian Q-function.
Since it is very difficult  to analyse  $\tilde{\mathbf{y}}_2(t-1)$ directly,
we approximate it using Eq. (\ref{eq:y_2_dif}) as
$
    \tilde{\mathbf{y}}_2(t)
 \approx  \sqrt{P_1 } \,S(t) \mathbf{h}_{12}(t)
 $.
This approximation is particularly accurate at high SNR.
Then, based on Eq. (\ref{eq:s_t_dif}), we have
$
S_{\Delta,ij}(t)
=U_{\Delta,ij}(t) S(t-1)
$
. We further assume
$
\mathbf{h}_{12}(t-1)\approx\mathbf{h}_{12}(t)
$
. Then, Eq. (\ref{eq:Q_function_PEP})
 can be further simplified as
\begin{equation}\label{eq:Q_function_PEP_diff}
\begin{split}
P_{kj}^{d}(\gamma)
&\approx \mathbb{E}_{f_i,g_i}\;Q\left(\sqrt {\frac{P_1    \|S_{\Delta,kj}(t) \mathbf{h}_{12}(t)\|^2 }{2 \sigma_{\tilde{\mathbf{n}}_2}^2(t)}}\right).
\end{split}
\end{equation}
Similarly, the PEP for the coherent scheme can be derived as
\begin{equation}\label{eq:Q_function_PEP_coh}
\begin{split}
P_{kj}^{c}(\gamma)
&=\mathbb{E}_{f_i,g_i} \;Q\left(\sqrt {\frac{P_1    \|S_{\Delta,kj}(t) \mathbf{h}_{12}(t)\|^2  }{2 \sigma_{\mathbf{n }_2}^2(t)}}\right).
\end{split}
\end{equation}

Since $ \sigma_{\tilde{\mathbf{n}}_2}^2(t) =2 \sigma_{\mathbf{n }_2}^2(t)$, the distributed differential scheme in TWRN  is supposed to have  $3$ dB loss in coding gain compared to distributed coherent scheme.

Before deriving the PEP, we first define
$
 \mathbf{h}_{12}(t)= \beta G(t)\hat{\mathbf{f}}(t)
$
, where
$
 \hat{\mathbf{f}}(t) =[\hat{f}_1(t), \cdots, \hat{f}_N(t)]^T
$
and
$
 G(t)=diag\{g_1(t), \cdots, g_N(t)\}
$. Then, we have the following lemmas.

\begin{lemma} \label{lemma:pdf_f}
The probability density function~(PDF) of $\hat{\mathbf{f}}(t)$ can be derived as
\begin{equation}\label{eq:pdf_f}
\begin{split}
p\left(\hat{\mathbf{f}}(t)\right) =\frac{1}{\pi^N  \sigma_{f}^{2N}}\exp\left(-\frac{\hat{\mathbf{f}}(t)^H \hat{\mathbf{f}}(t)}{\sigma_{f}^{2}}\right).
\end{split}
\end{equation}
\end{lemma}
\begin{proof}[Proof]
Since $f_i(t) \sim \mathcal{CN}(0, \sigma_{f}^{2})$, we can prove  that $f_i^{*}(t) \sim \mathcal{CN}(0, \sigma_{f}^{2})$. Hence, $\hat{f}_i(t) \sim \mathcal{CN}(0, \sigma_{f}^{2})$. Note that $\hat{f}_1(t), \cdots, \hat{f}_N(t)$ are independent, we can easily derive Eq.~(\ref{eq:pdf_f}).
\end{proof}

\begin{lemma}\label{lemma:pdf_int}
 $B$ represents an $n\times n$ Hermitian matrix (i.e., $B^H=B$), and  $\mathbf{x}$ is an $n\times 1$ complex vector. We have
\begin{equation}\label{}
\begin{split}
 \int_{\mathcal{C}^n}  \exp\left(- \mathbf{x}^H B \mathbf{x} \right) \text{d}\mathbf{x}
 = \pi^{n}\text{det}^{-1}(B).
\end{split}
\end{equation}
\end{lemma}
\begin{proof}[Proof]
Please see \cite{Dogandzic2003P1327}.
\end{proof}

Note that the canonical representation of Gaussian Q-function is in the form of a semi-infinite integral, which makes analysis very difficult.  Here, we use an alternative representation of the Gaussian Q-function from \cite[Eq. (4.2)]{Simon2005P} as
$ Q(x)
 = \frac{1}{\pi}\int_0^{\pi/2}
 \exp\left(-\frac{x^2}{2sin^2\theta}\right)\text{d}\theta
$.

Then, by doing some manipulations, we have
\begin{equation}\label{eq:Qfunc_PEP}
\begin{split}
P_{kj}^{d}(\gamma)
&= \mathbb{E}_{f_i,g_i}\;\frac{1}{\pi}\int_0^{\pi/2}\exp
\left[-\frac{\hat{\mathbf{f}}(t)^H K(t) \hat{\mathbf{f}}(t)  }{2 \sin^2\theta}\right]
\text{d}\theta
\\
&= \frac{1}{\pi}\int_0^{\pi/2}\mathbb{E}_{g_i}\;   \left[ \text{det}(\mathbf{I} + K^{\prime}(\theta,t) \right]^{-1}    \text{d}\theta
\\
&= \frac{1}{\pi}\int_0^{\pi/2}\mathbb{E}_{g_i}\;
\left[\prod_{i=1}^{N} \left(1+l(\theta,t)
 \lambda_i  |g_i(t)|^2\right)
\right]^{-1}    \text{d}\theta,
 \end{split}
\end{equation}
where
$
K(t)
=\frac{P_1 |\beta|^2      G(t)^H   S_{\Delta,kj}(t) ^H  S_{\Delta,kj}(t) G(t)   }{4( \sum_{i=1}^{N} |\beta|^2 |g_i(t)|^2 +1 )  N_0  }
$,
$
K^{\prime}(\theta, t)
=\frac{ \sigma_{f}^{2} K(t)}{2 \sin^2\theta}
$,
$
l(\theta, t)
 =\frac{P_1 |\beta|^2  \sigma_{f}^{2}  }{8( \sum_{i=1}^{N} |\beta|^2 |g_i(t)|^2 +1 )  N_0 \sin^2\theta}
$,
and
$\lambda_i$,  $i\in \{1, \cdots, N\}$, denotes the singular value of $S_{\Delta,kj}(t) ^H  S_{\Delta,kj}(t)$.
The second step of the equation is based on the Lemma \ref{lemma:pdf_f} and Lemma \ref{lemma:pdf_int}. 

Note that the mean of $|g_i(t)|^2$ is $\sigma_{g}^{2}$. It is reasonable to approximate the term $\sum_{i=1}^{N} |g_i(t)|^2$ in $l(\theta, t)$, by $\sum_{i=1}^{N} |g_i(t)|^2 \approx N \sigma_{g}^{2}$, especially for large $N$ (by the law of large numbers)\cite{Jing2006P3524,Maham2009P2036}. Hence,
\begin{equation}\label{eq:approx_l_theta}
\begin{split}
l(\theta, t)
&\approx l^{\prime}(\theta)=\frac{P_1 |\beta|^2  \sigma_{f}^{2} }{8( N |\beta|^2 \sigma_{g}^{2} +1 )  N_0 \sin^2\theta}.
 \end{split}
\end{equation}

Let $|g_i(t)|^2=\gamma_i(t)$. Since  $g_i \sim \mathcal{CN}(0, \sigma_{g}^{2})$, the PDF of $\gamma_i(t)$ can be obtained as $p\left(\gamma_i(t)\right) =\frac{1}{\sigma_{g}^2}\exp\left(-\frac{\gamma_i(t)}{\sigma_{g}^2}\right).$   
 Hence, after doing some manipulations, the MGF-based 
PEP expression is derived as
\begin{equation}\label{eq:closed-form PEP}
\begin{split}
P_{kj}^{d}(\gamma)
&\approx \frac{1}{\pi}\int_0^{\pi/2}\mathbb{E}_{g_i}\;
\left[\prod_{i=1}^{N} \left(1+l^{\prime}(\theta)
 \lambda_i  |g_i(t)|^2\right)
\right]^{-1}    \text{d}\theta
\\
 &=\frac{1}{\pi} \int_0^{\pi/2}
\prod_{i=1}^{N}
\left[
-(\frac{\sin^2\theta}{M_i }) \exp\left({ \frac{\sin^2\theta}{M_i }}\right) \right.\\
&\left. \qquad \qquad \qquad \qquad \qquad \quad \times \mathbf{Ei}\left(-\frac{\sin^2\theta}{M_i  }\right)
\right]
 \text{d}\theta,
 \end{split}
\end{equation}
where
$
M_i
=\frac{P_1 |\beta|^2  \sigma_{f}^{2}\sigma_{g}^{2} }{8( N |\beta|^2 \sigma_{g}^{2} +1 )  N_0 }\lambda_i
$
and
$
\mathbf{Ei}\left(x\right)=\int_{-\infty}^{x}\frac{e^t}{t}\text{d}t 
$, for $x<0$,
is the exponential integral function\cite[8.211.1]{Gradshteyn2007P}.

Next let us   derive the simplified PEP expression at high SNR.
Note that\cite[8.214.1]{Gradshteyn2007P} $\mathbf{Ei}\left(x\right)=\mathbf{C}+\ln(-x)+\sum_{k=1}^{\infty}\frac{x^k}{k\cdot k!}, \; [x<0],$
where $\mathbf{C}$ is Euler’s constant and $\mathbf{C}\approx 0.577$\cite[9.73]{Gradshteyn2007P}.
When $x$ tends to $0$, the exponential integral function
can be approximated as
$
\mathbf{Ei}\left(x\right)\approx  \ln(-x)
$, for $x<0$.
At high SNR, we have
$
\exp\left({ \frac{\sin^2\theta}{M_i  }}\right) \approx 1
$
and using the approximation for  the exponential integral function, we have
\begin{equation}\label{}
\begin{split}
 P_{kj}^{d}(\gamma)
   \approx\frac{1}{\pi} \int_0^{\pi/2}
\prod_{i=1}^{N}
&\left[
-(\frac{\sin^2\theta}{M_i }) \right. \\
&\qquad \left. \times
\left(
2\ln\left( \sin \theta \right)
+  \ln\left( \frac{1}{M_i  }\right)
\right)
\right]
 \text{d}\theta.
 \end{split}
\end{equation}
Note that $\int_0^{\frac{\pi}{2}}\ln \sin x\,\text{d}x=-\frac{\pi}{2}\ln 2$\cite[4.224.3]{Gradshteyn2007P}.
Hence the $\ln (\sin \theta)$ can be ignored, especially at high SNR.
Using\cite[3.621.3]{Gradshteyn2007P},  we  have $\int_0^{\frac{\pi}{2}} \sin^{2m}x \text{d}x=\frac{(2m-1)!!}{(2m)!!}\frac{\pi}{2}$.
The PEP can be further simplified as
\begin{equation}\label{eq:simplified PEP}
\begin{split}
P_{kj}^{d}(\gamma)
&\lessapprox
 \frac{1}{2} \frac{(2N-1)!!}{(2N)!!}
\prod_{i=1}^{N}
\left[
 (\frac{1}{M_i })
\ln\left(  M_i  \right)
\right].
\end{split}
\end{equation}

Finally, we derive the well-known Chernoff-bound-based PEP expression. From Eq. (\ref{eq:Qfunc_PEP}), setting $\theta=\frac{\pi}{2}$, and doing some manipulations, the Chernoff-bound-based PEP expression is given as
\begin{equation}\label{eq:cherrnoff}
\begin{split}
P_{kj}^{d}(\gamma)
&\leqslant
\frac{1}{2} \,\mathbb{E}_{f_i,g_i}\;  \exp
\left[-\frac{\hat{\mathbf{f}}(t)^H K(t) \hat{\mathbf{f}}(t)  }{2 }\right]\\
 &=
 \frac{1}{2}
\prod_{i=1}^{N}
\left[
 (\frac{1}{M_i })
\ln\left(  M_i  \right)
\right].
 \end{split}
\end{equation}

The average BLER can be obtained based on the well-known union bound as
\begin{equation}\label{eq:FER}
\begin{split}
P_{BLER} ^{d}(\gamma) \leqslant \sum_{U_k\in\mathcal {U}} \sum_{\; j,\; j \neq k} Pr(U_k)P_{kj}^{d}(\gamma).
\end{split}
\end{equation}

\subsection{Diversity Order}\label{subsec:diversity}
In this subsection, we analyze the asymptotic diversity order of the proposed DDSTC-ANC scheme.
Firstly, we define the total transmission power is $N \cdot P$. Note that $N\cdot P=N\cdot P_1+N\cdot P_2+N^2\cdot P_{Ri}$, $P_1=\alpha_1 P$, and $P_2=\alpha_2 P$.  Denote the SNR $\gamma=\frac{P}{N_0}$. Then, we rewrite
$M_i$ at high SNR as
$M_i=C \lambda_i\, \gamma$,
where
$
C =\frac{\alpha_1 \frac{(1-\alpha_1-\alpha_2)}{N} \sigma_{f}^{2}\sigma_{g}^{2} }{8( (1-\alpha_1-\alpha_2) \sigma_{g}^{2} +\alpha_1\sigma_{f}^{2}  +\alpha_2\sigma_{g}^{2}  +1/\gamma  )     }
\approx \frac{\alpha_1 \frac{(1-\alpha_1-\alpha_2)}{N} \sigma_{f}^{2}\sigma_{g}^{2} }{8( (1-\alpha_1-\alpha_2) \sigma_{g}^{2} +\alpha_1\sigma_{f}^{2}  +\alpha_2\sigma_{g}^{2}    )     }
$.
Thus, the simplified PEP at high SNR can be rewritten as
 \begin{equation}\label{ }
\begin{split}
  P_{kj}^{d}(\gamma)
  &\approx
  \frac{1}{2} \frac{(2N-1)!!}{(2N)!!}
  \frac{1}{ \prod_{i=1}^{N} C  \lambda_i}
    \gamma ^{-N}   \prod_{i=1}^{N}\left( \ln(C  \lambda_i)+\ln(\gamma ) \right)\\
  &\approx  (C^{\prime}  \,\gamma )^{-N} \left[\ln(\gamma ) \right]^N,
\end{split}
\end{equation}
where $C^{\prime}  =\left(\frac{1}{2} \frac{(2N-1)!!}{(2N)!!}
\frac{1}{ \prod_{i=1}^{N} C \lambda_i}\right)^{-\frac{1}{N}}$.
 When $S_{\Delta,kj}(t)S_{\Delta,kj}(t)^H$ is full rank, the diversity can be  obtained as\cite{Zheng2003P1073}
 \begin{equation}\label{ }
\begin{split}
d=\lim_{\gamma \rightarrow\infty}-\frac{\log(P_{k,j}^d(\gamma))}{\log(\gamma)}=N\left(1-\frac{\log\log(\gamma)}{\log(\gamma)}\right).
\end{split}
\end{equation}
Thus, the diversity of the proposed DDSTC-ANC scheme for TWRNs
is $N\left(\frac{1-\log\log(\gamma)}{\log(\gamma)}\right)$.


\section{Optimum Power Allocation}\label{sec:OPA}

In this section, we derive the OPA between the source nodes and the relay nodes that minimizes the total PEP in  the TWRNs. Because the MGF-based  PEP expression is very hard to analyze and gives little insight, we use the simplified PEP expression to derive the OPA.    Here, we consider the total PEP in the TWRNs, and denote the PEP in $T_1$ and $T_2$ as $P_{ij}^{d,1}(\gamma)$ and $P_{ij}^{d,2}(\gamma)$, respectively.  $C$ in Subsection \ref{subsec:diversity} is rewritten as $C_{T_1}$ and $C_{T_2}$ for  $T_1$ and $T_2$, respectively. Hence, we have
   \begin{equation}\label{eq:total PEP}
\begin{split}
  P_{ij}^{d,1}(\gamma)+  P_{ij}^{d,2}(\gamma)
      &\approx
  \frac{1}{2} \frac{(2N-1)!!}{(2N)!!}\frac{1}{ \prod_{i=1}^{N}    \lambda_i}\\
  &\qquad \times
    \left(C_{T_1}^{-N}+C_{T_2}^{-N}\right)
    \gamma ^{-N} \left[\ln(\gamma ) \right]^N,
\end{split}
\end{equation}
where
\begin{equation*}C_{T_1} \approx \frac{\alpha_2 \frac{(1-\alpha_1-\alpha_2)}{N} \sigma_{f}^{2}\sigma_{g}^{2} }{8( (1-\alpha_1-\alpha_2) \sigma_{f}^{2} +\alpha_1\sigma_{f}^{2}  +\alpha_2\sigma_{g}^{2}    )     }\end{equation*}
and
\begin{equation*}C_{T_2} \approx \frac{\alpha_1 \frac{(1-\alpha_1-\alpha_2)}{N} \sigma_{f}^{2}\sigma_{g}^{2} }{8( (1-\alpha_1-\alpha_2) \sigma_{g}^{2} +\alpha_1\sigma_{f}^{2}  +\alpha_2\sigma_{g}^{2}    ).     }\end{equation*} 
It is obvious that to minimize the PEP at high SNR, we should minimize the $C_{T_1}^{-N}+C_{T_2}^{-N}$ in Eq. (\ref{eq:total PEP}) .i.e.,
 \begin{equation}\label{eq:OPA}
\begin{split}
 &\min_{\alpha_1,\alpha_2}\{C_{T_1}^{-N}+C_{T_2}^{-N}\},
 \qquad \text{s.t.}
 \begin{cases}
  &  \alpha_1  \geqslant 0  , \; \; \; \alpha_2 \geqslant 0,  \\
  &  \alpha_1 +\alpha_2 \eqslantless 1. \\
  \end{cases}
\end{split}
\end{equation}
As a special case, when $\sigma_{f}^{2}=\sigma_{g}^{2}=\sigma^{2}$, we have $\alpha_1=\alpha_2=\alpha$. Therefore,
 \begin{equation}\label{}
\begin{split}
C_{T_1}=C_{T_2}=\frac{2\alpha(1-2\alpha)\sigma^2}{16N}\leq \frac{\sigma^2}{64N},
\end{split}
\end{equation}
with equality when $\alpha=\frac{1}{4}$, or equivalently, $P_1=P_2=\frac{P}{4}$ and $P_{R_i}=\frac{P}{2N}$.
Thus, the OPA is such that the source nodes use half the total power and the relay nodes share the other half.
\emph{We should emphasize} that this power allocation only works
for the TWRNs, in which
all channels are assumed to be i.i.d. Rayleigh and no path-loss
is considered. It is obvious that it may not be optimal when
the path-loss effect is considered in the TWRNs.

As the expression in Eq. (\ref{eq:OPA}) is complicated, it is difficult to derive the closed-form solution for OPA when $\sigma_{f}^{2}\neq\sigma_{g}^{2}$. Here, we use numerical method, such as the nonlinear optimization method, to obtain the optimal solution.
In Section \ref{sec:simulations}, it is interesting to find that when $\sigma_{f}^{2}\neq\sigma_{g}^{2}$, $\alpha_1+\alpha_2=0.5$ still holds  for the simulated scenarios,  which means the source nodes still share half the total power.





\section{Simulations}\label{sec:simulations}

\begin{figure}[]
\centering
\includegraphics[width=0.75\textwidth]{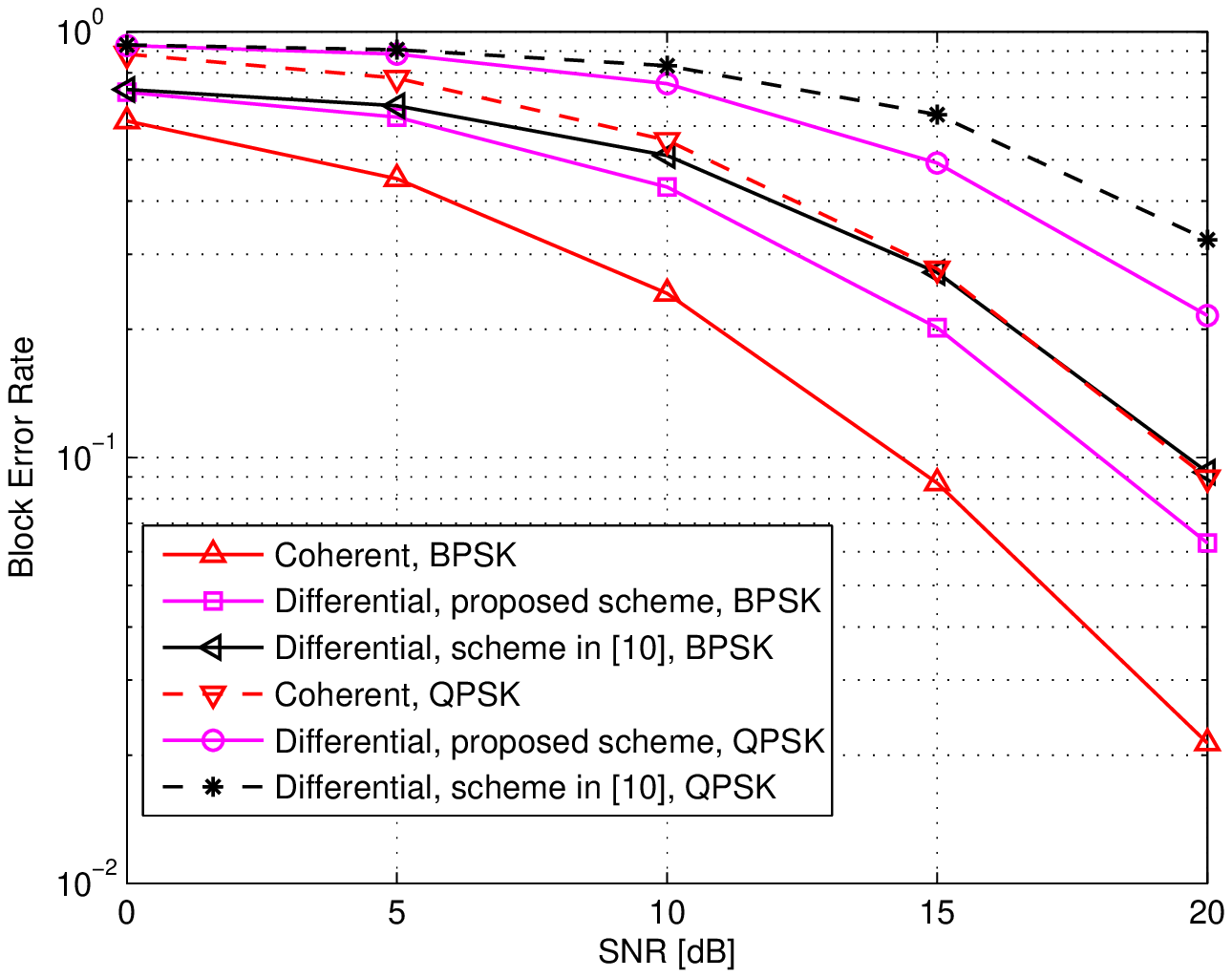}
\caption{\small Simulated BLER performance using Alamouti codes~($2$ relays).}  \label{fig:Simulation_Alamouti_BPSK_4PSK_Coh_Dif}
\end{figure}


\begin{figure}[]
\centering
\includegraphics[width=0.75\textwidth]{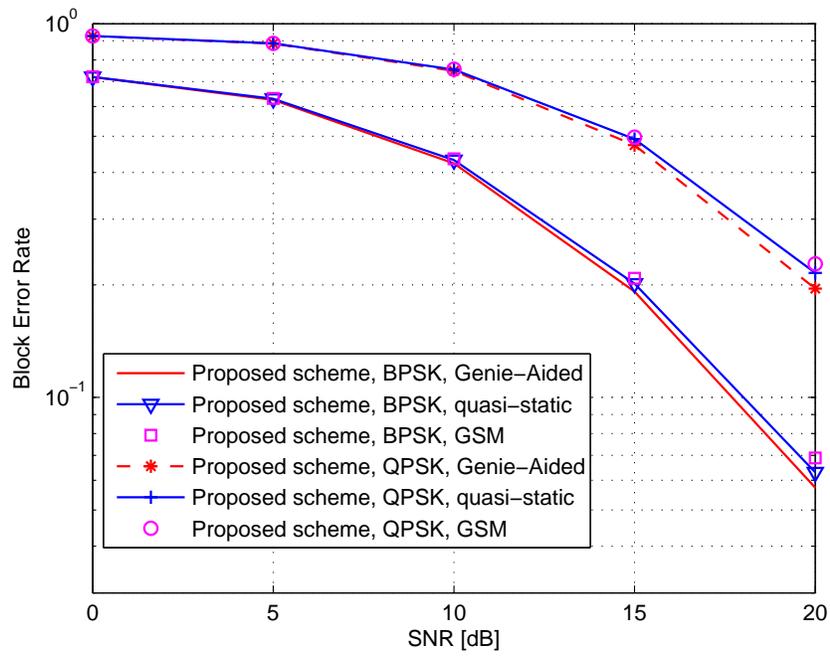}
\caption{\small Simulated BLER performance using Alamouti codes~($2$ relays) over a GSM channel and quasi-static Rayleigh fading channel.}
\label{fig:Simulation_Alamouti_BPSK_4PSK_Coh_Dif_GSMchan}
\end{figure}



\begin{figure}[]
\centering
\includegraphics[width=0.75\textwidth]{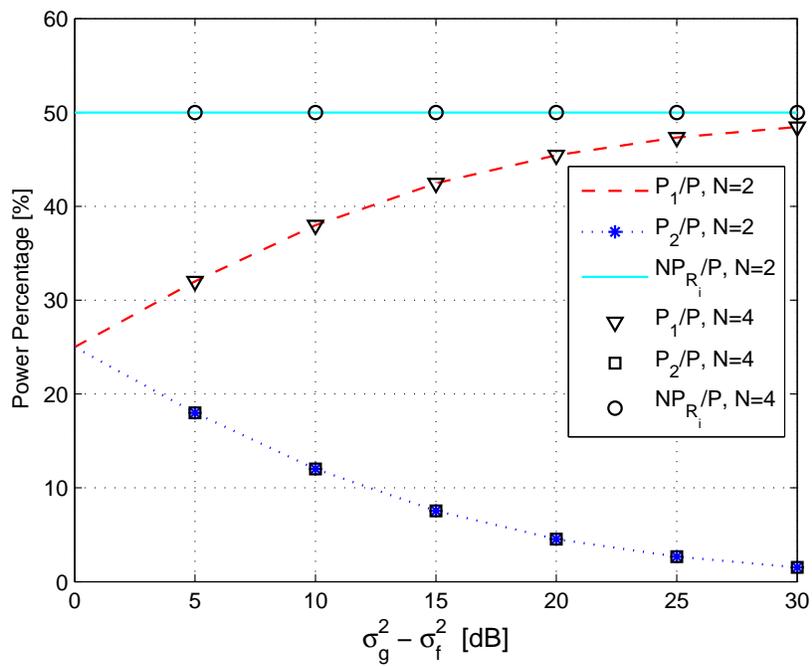}
\caption{ \small Optimum power allocation between source and relay nodes.
}  \label{fig:Theory_PowerAllocation}
\end{figure}


\begin{figure}[]
\centering
\includegraphics[width=0.75\textwidth]{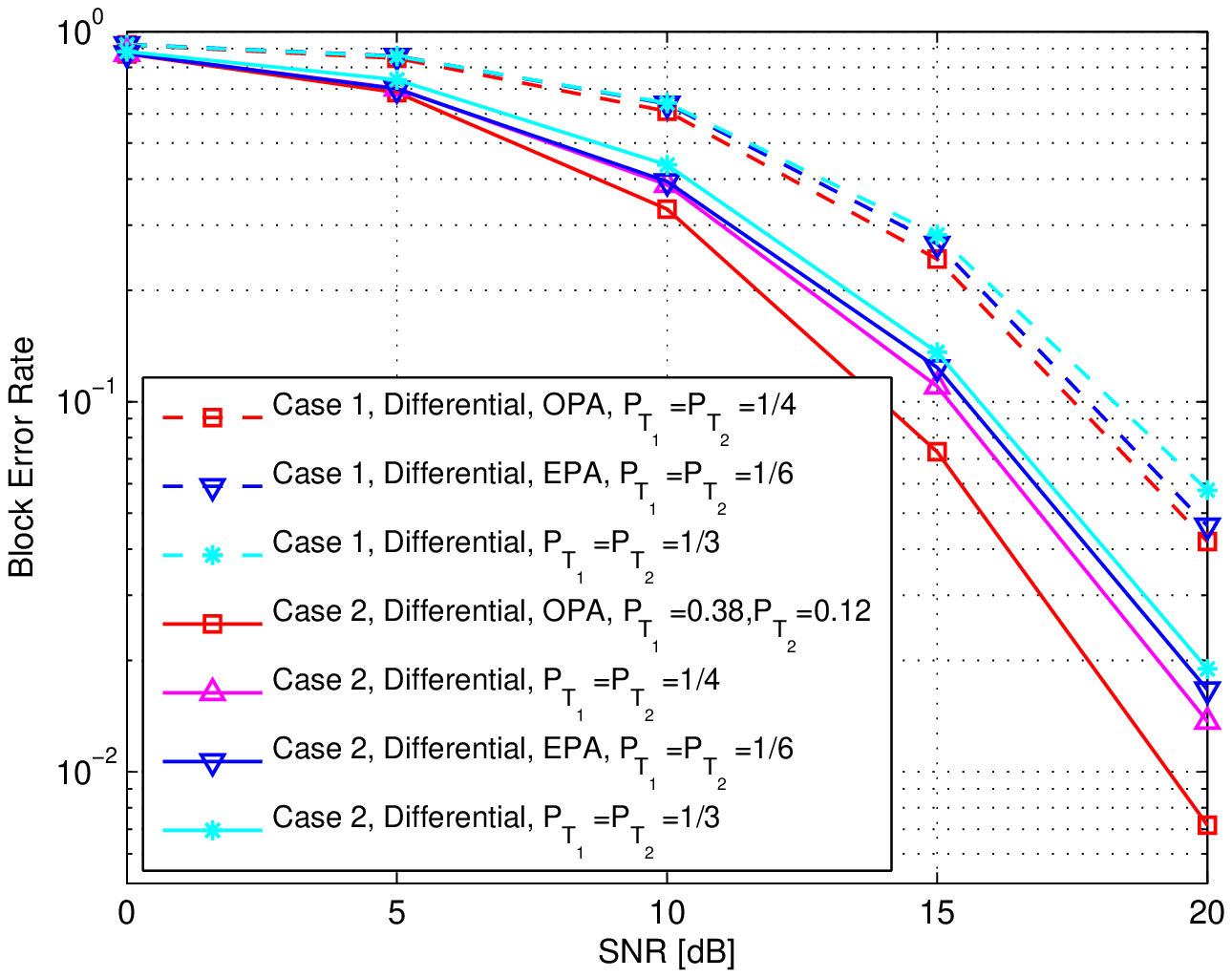}
\caption{ \small Simulated BLER performance by the proposed DDSTC-ANC using SORC with transmit power allocation~($4$ relays).
}  \label{fig:Simulation_SORC_Dif_EPA_OPA}
\end{figure}

In this section, we provide simulation results for the proposed
DDSTC-ANC scheme.
Simulations are
performed with PSK modulation and a frame size of $100$ symbols over a quasi-static Rayleigh fading
channels  without specific mention.
The estimated $\mathbf{h}_{11}(t)$ and $\mathbf{h}_{22}(t)$ are used in simulations.
For comparison, we also present simulations over a GSM channel model with a symbol sampling period of $T_s=3.693 \mu s$ and a maximum Doppler shift of $75$ Hz\cite{Utkovski2009P779}. This ensures a slowly changing channel and allows the assumption of a constant channel over two consecutive time blocks.  Without specific mention, we assume that $\sigma_{f}^2=\sigma_{g}^2=1$ and the source nodes uses half the total power and the relay nodes share the other half, i.e., $P_1=\frac{1}{4}P$, $P_2=\frac{1}{4}P$ and $P_{R_i}=\frac{1}{2N}P$.

From Fig. \ref{fig:Simulation_Alamouti_BPSK_4PSK_Coh_Dif},
we present the simulated BLER performance for the proposed DDSTC-ANC schemes using Alamouti  for TWRNs.
The performance of the corresponding coherent detection is plotted as well for better comparison.
It shows that the differential
scheme suffers about $3$-dB performance loss compared with the
corresponding coherent scheme, which has been validated in Subsection \ref{subsec:PEP}.
Fig. \ref{fig:Simulation_Alamouti_BPSK_4PSK_Coh_Dif}  also compares the simulated BLER performance for our proposed DDSTC-ANC and the differential scheme in \cite{Utkovski2009P779}.
It can be observed that our proposed scheme is superior to  (about $2$-dB)   the detector in \cite{Utkovski2009P779}.
The main reason is that the differential detection approach employed in \cite{Utkovski2009P779}
was based on the estimation of the previous symbol. Consequently, when one symbol was decoded incorrectly, it will affect the decoding of the consecutive symbols thus leading to   serious error propagation.
Comparatively, the information about the estimation of the previous symbol is not required in our proposed differential detection and is, thus, able to prevent the error propagation.

In Fig. \ref{fig:Simulation_Alamouti_BPSK_4PSK_Coh_Dif_GSMchan}, we include the Genie-aided results by assuming that  each source node can  perfectly remove its own information from the received signal.
It can be noted from the results that
the proposed differential detection scheme introduces negligible performance loss compared to the genie-aided scheme.  
We also compare the BLER performance of the differential scheme over a GSM channel~(a practical channel) and a quasi-static Rayleigh fading channel.
From the figure, it can be observed that there is almost no performance loss in a GSM channel compared to the quasi-static Rayleigh fading channel which clearly justifies the robustness of the proposed differential scheme in slow fading channels.
It also indicates that  the effect of non-constant channel on proposed scheme can be ignored which validate our assumption of quasi-static fading channel model.

%

In Fig.  \ref{fig:Theory_PowerAllocation}, we show the optimum power allocation scheme of the DDSTC-ANC scheme.  It can be seen that more  power should be allocated to $P_1$ when the channels from relay nodes to $T_2$ are  better than the channels from relay nodes to $T_1$. It is interesting to find that when $\sigma_{f}^{2}\neq\sigma_{g}^{2}$, the sources still share half the total power for the optimal power allocation.

In  Fig. \ref{fig:Simulation_SORC_Dif_EPA_OPA},
we examine the BLER performance of the proposed scheme with power allocation for the system with four relay nodes. The SORC is used at relays and signal is modulated from a BPSK constellation.
We also take into account the relay's location as: case 1~(the \emph{symmetric} case), where relays are placed halfway between the source nodes, i.e., $T_1,T_2$, and $\sigma_{f}^2=1$ and $\sigma_{g}^2=1$; and case 2~(the \emph{asymmetric} case), where relays are close to the source node~$T_2$, and $\sigma_{f}^2=1$ and $\sigma_{g}^2=10$.
It can be observed   from   Fig. \ref{fig:Simulation_SORC_Dif_EPA_OPA} that the BLER performance of the proposed scheme with power allocation
can provide considerable performance gain in comparison with the equal power allocation~(EPA) scheme, i.e., $P_1=P_2=P_{R_i}=\frac{P}{N+2}$.



\section{Conclusion }\label{sec:conclusion}

In this paper, we have proposed a DDSTC-ANC scheme for TWRNs with multiple relays.
A simple differential signal detector was developed to recover the desired signal at each source   by subtracting its contribution from the broadcasted signals.
The performance of the proposed DDSTC-ANC scheme was analyzed and the OPA was presented to improve the system performance.
Analytical results have been verified through Monte Carlo simulations.











%




\bibliographystyle{IEEEtran}
\bibliography{IEEEabrv,DD_STC_ANC}

\end{document}